\documentclass[11pt]{amsart}
\usepackage[a4paper,margin=1in]{geometry}
\usepackage{amsmath,amssymb,amsthm}
\usepackage{tikz}
\usepackage{xcolor}
\usepackage[hidelinks]{hyperref}
\newcommand{\cyclepath}{revision/}
\usepackage{bm} 
\usepackage{url}
\usepackage{comment}
\usepackage{mathtools}

\usepackage[colorinlistoftodos]{todonotes}

\usepackage{enumitem}
\usepackage{subcaption}

\usepackage{microtype}

\newcommand{\phinorm}[1]{\|#1\|_{\Phi}}
\newcommand{\norm}[1]{\|#1\|}
\newcommand{\E}{\mathop{\mathbb E}}    

\newcommand{\calH}{\mathcal H}
\newcommand{\calV}{[n]}
\newcommand{\calE}{\mathcal E}

\theoremstyle{plain}                   
\newtheorem{thm}{Theorem}[section]
\newtheorem{lmm}[thm]{Lemma}
\newtheorem{prop}[thm]{Proposition}
\newtheorem{cor}[thm]{Corollary}
\newtheorem{defn}[thm]{Definition}

\theoremstyle{definition}              

\newtheorem{conj}[thm]{Conjecture}

\newtheorem{rmk}[thm]{Remark}

\allowdisplaybreaks
\newcommand{\revision}[1]{#1}
\title{Listing Even Cycles Faster than the Submodular-Width Barrier}
\author{Vasileios Nakos}
\address{National and Kapodistrian University of Athens, Athens, Greece}
\email{vasilisnak@di.uoa.gr}
\author{Hung Ngo}
\address{RelationalAI Inc., Berkeley, CA, USA}
\email{hung.q.ngo@gmail.com}
\author{Andreas Panayi}
\address{National and Kapodistrian University of Athens, Athens, Greece}
\email{andreaspanayi8@gmail.com}
\begin{document}
\begin{abstract}
A classic result of Alon, Yuster, and Zwick (AYZ, Algorithmica 1997) shows that all
$2k$-cycles in an $m$-edge (directed or undirected) graph can be listed in $\tilde
O(m^{2-1/k} + t)$ time, where $t$ is the output size. This bound is a starting point for the
notion of {\em submodular width} due to Marx (JACM, 2013), and then the PANDA framework by Abo
Khamis, Ngo, and Suciu (PODS, 2017), which generalize the AYZ result to arbitrary
conjunctive queries and input degree constraints. A central open question is whether
combinatorial algorithms can beat the submodular-width barrier.

Bringmann and Gorbachev (STOC 2025) gave lower-bound evidence that submodular width may be
optimal for general conjunctive queries under combinatorial algorithms. The picture changes,
however, for $2k$-cycles on undirected graphs, whose corresponding queries have self-joins
and symmetric input \revision{relations}: recent works have improved on AYZ for even-cycle detection and
listing. Pinning down the complexity of \revision{$2k$-cycle} detection and listing is thus a natural
step toward overcoming the submodular-width barrier for queries with self-joins or
symmetric \revision{input relations}.

With respect to the detection problem (a Boolean conjunctive query), Dahlgaard,
Knudsen, and St{\"{o}}ckel (STOC 2017) showed that \revision{$2k$-cycle} detection can be solved in
$\tilde O(m^{2k/(k+1)})$ time, improving on the AYZ bound. The even-cycle listing problem
(a full conjunctive query) is more challenging. Recent works have made progress
on small cycles. Jin and Xu (STOC 2023) (and independently Abboud, Khoury, Leibowitz, and
Safier, FSTTCS 2023) showed that 4-cycles can be listed in $\tilde O(m^{4/3}+t)$ time, and
Vassilevska~Williams and Westover (ITCS 2025) showed that 6-cycles can be listed in $\tilde
O(m^{8/5}+t)$ time; both results improve the corresponding AYZ bounds of $\tilde O(m^{3/2})$
and $\tilde O(m^{5/3})$, respectively. The general case, however, has remained open since
the original AYZ paper some 30 years ago.

In this paper, building on the works of Dahlgaard, Knudsen, and St{\"{o}}ckel, and of
Vassilevska~Williams and Westover, we prove that $2k$-cycles can be listed in $\tilde
O\left(m^{(2k^2-k+1)/(k^2+1)}+t\right)$ time, improving the AYZ bound for every $k \geq 3$.
The key technical ingredient is an \emph{asymmetric supersaturation} result for even cycles.
Additionally, our algorithms are expressed entirely in terms of \revision{explicit
relational operators} over multiple tree-decomposition query plans -- \revision{which may
facilitate their implementation in database systems.} This is in contrast to the
breadth-first search (BFS)-based approaches in the prior graph algorithms literature.
\end{abstract}

\maketitle
\section{Introduction}
Conjunctive query (CQ) evaluation is a fundamental problem across many areas, including
algorithmic graph theory, relational database management, constraint satisfaction, and logic
programming~\cite{DBLP:conf/pods/KhamisNR16, DBLP:conf/pods/GottlobGLS16,
DBLP:journals/talg/GroheM14, DBLP:books/aw/AbiteboulHV95,theoretics:13722}. There is a
natural synergy between graph algorithms and query evaluation. In particular, a graph is
simply a binary relation, and subgraph pattern queries can naturally be expressed as CQs;
thus, advances in graph algorithms have often translated directly into progress on query
evaluation, and vice versa.

For example, a special case of conjunctive query evaluation is the problem of finding or
listing a small (hyper)graph pattern inside a large graph. A long line of work on this
problem~\cite{DBLP:conf/stoc/Itai77,MR599482,
MR859293,MR1639767,MR2104047,DBLP:conf/soda/GroheM06,
DBLP:journals/talg/GroheM14,DBLP:journals/siamcomp/AtseriasGM13} led to the {\em
AGM-bound}~\cite{DBLP:journals/siamcomp/AtseriasGM13} for the worst-case output size of
conjunctive queries, and the notion of {\em worst-case optimal join
algorithms}~\cite{DBLP:journals/sigmod/NgoRR13, DBLP:conf/pods/NgoPRR12,
DBLP:conf/icdt/Veldhuizen14, 10.1145/3196959.3196990} for evaluating CQs in time
proportional to the worst-case output size.

Another example that is central to our story is the classic result of Alon, Yuster, and
Zwick~\cite{AYZ97}: given any (directed or undirected) graph
with $m$ edges and a constant $k$, one can detect whether a $k$-cycle exists in time $\widetilde{O}(m^{2-1/\lceil
k/2 \rceil})$, and list all $k$-cycles in time $\widetilde{O}(m^{2-1/\lceil k/2 \rceil} + t)$,
where $t$ is the number of $k$-cycles. Marx~\cite{DBLP:journals/jacm/Marx13} introduced the
notion of \emph{submodular width} of a conjunctive query; the PANDA
framework~\cite{theoretics:13722} later extended this to a general algorithm that evaluates
any conjunctive query $Q$ over a database of size $N$ in time
$\widetilde{O}(N^{\mathsf{subw}(Q)} + t)$, where $\mathsf{subw}(Q)$ is the submodular width of
$Q$ and $t$ is the output size. The submodular width of the $k$-cycle query is $2-1/\lceil k/2
\rceil$, and thus we recover the AYZ result as a special case.

There is some evidence that the submodular width may be the best one can hope for in
general, amongst ``combinatorial algorithms''. In particular, Bringmann and
Gorbachev~\cite{DBLP:conf/stoc/BringmannG25} recently proved that for the class of so-called
\emph{sub-quadratic} graph patterns --- of which cycles are a special case --- the
submodular-width runtime is indeed optimal, assuming the {\sf
MinClique}~\cite{VassilevskaWilliams2018ICM}, {\sf
3SUM}~\cite{DBLP:journals/comgeo/GajentaanO12}, or {\sf
ZeroClique}~\cite{DBLP:conf/stoc/AbboudBDN18} hypotheses.

Before proceeding, it is worth clarifying precisely what we mean by a $k$-cycle query.
Let $E$ be the edge relation of a (directed or undirected) graph, and an integer constant $k\geq 2$.
The $k$-cycle \emph{detection} query is the Boolean conjunctive query
\begin{align}
    C_k() \leftarrow E(v_1,v_2) \wedge E(v_2,v_3) \wedge \cdots \wedge E(v_{k-1},v_k)
    \wedge E(v_k,v_1) \wedge \bigwedge_{i \neq j} v_i \neq v_j,
    \label{eqn:ck:detection}
\end{align}
and the $k$-cycle \emph{listing} query outputs the cycles:
\begin{align}
    C_k(v_1,\ldots,v_k) \leftarrow E(v_1,v_2) \wedge E(v_2,v_3) \wedge \cdots
    \wedge E(v_{k-1},v_k) \wedge E(v_k,v_1) \wedge \bigwedge_{i \neq j} v_i \neq v_j.
    \label{eqn:ck:listing}
\end{align}
These are crucially different from the $k$-\emph{closed-walk} queries, which drop the
distinctness constraints:
\begin{align}
    W_k() &\leftarrow E(v_1,v_2) \wedge E(v_2,v_3) \wedge \cdots \wedge E(v_{k-1},v_k)
    \wedge E(v_k,v_1),
    \label{eqn:wk:detection}
    \\
    W_k(v_1,\ldots,v_k) &\leftarrow E(v_1,v_2) \wedge E(v_2,v_3) \wedge \cdots
    \wedge E(v_{k-1},v_k) \wedge E(v_k,v_1).
    \label{eqn:wk:listing}
\end{align}
\revision{When expressed as a conjunctive query~\eqref{eqn:ck:listing}, the vertices on the
cycle are named explicitly as the variables $v_1, \dots, v_k$, and thus the output size
$|C_k|$ is $k$-times the number of $k$-cycles in the graph. Since
$k$ is a constant in our context, the $k$-cycle listing problem in the graph context and the
conjunctive query context are equivalent, up to a constant factor.}

The closed-walk queries are standard CQs and fall directly within the scope of the PANDA
framework, whereas the cycle queries require handling the $\neq$ constraints, which are
harder to deal with. In particular, variables in the cycle queries are fully connected into
a clique, which means there is no non-trivial tree decomposition. Applying an algorithm such
as PANDA to the cycle queries amounts to brute-force search over all $k$-tuples of vertices,
which takes $O(n^k)$ time. We will explain in Section~\ref{sec:background} how the color-coding
technique of Alon, Yuster, and Zwick~\cite{color-coding} can be used to handle the $\neq$
constraints via Boolean tensor decomposition~\cite{DBLP:conf/icdt/KhamisNOS19}. In
particular, the submodular width of the $k$-closed-walk queries is $2-1/\lceil k/2 \rceil$;
together with the color-coding technique, this lets us evaluate the $k$-cycle queries in
the time predicted by the submodular width.

A key open question in algorithmic graph theory over the past 30 years has been to beat the
runtime of the AYZ algorithm. Recently, there has been some progress on the \revision{$2k$-cycle}
detection and listing problems in \emph{undirected} graphs. Before describing this progress,
it is worth clarifying upfront how ``beating the AYZ runtime'' is possible at all in light of
the recent fine-grained lower bound of Bringmann and
Gorbachev~\cite{DBLP:conf/stoc/BringmannG25}, which matches the AYZ exponent. Their bound
applies to the \emph{colored} variant of cycle listing, in which the host graph's vertex set
$V$ is partitioned into $2k$ classes and each of the $2k$ atoms of the cycle query is governed
by its own bipartite relation $E_1, \dots, E_{2k}$ between consecutive classes — these $2k$
relations may be chosen independently and adversarially. Our setting is strictly more
restricted: all $2k$ atoms refer to a single shared edge relation $E$, and $E$ is symmetric
because the host graph is undirected. Both restrictions matter — the shared-relation
property already rules out the adversarial 3SUM-style gadgets
of~\cite{DBLP:conf/stoc/BringmannG25}, and symmetry of $E$ is what enables the
supersaturation tools that our algorithm is built on.

On the detection side, Yuster and Zwick~\cite{DBLP:journals/siamdm/YusterZ97} showed how to
detect a $2k$-cycle in $O(n^2)$-time for all $k$, where $n$ is the number of vertices of the
graph. In the dense regime, this is believed to be the best possible. In the sparse regime,
Dahlgaard, Knudsen, and St{\"{o}}ckel~\cite{DBLP:conf/stoc/DahlgaardKS17} showed how to
detect a $2k$-cycle in an {\em undirected} graph in time $O(m^{2k/(k+1)})$, improving on
the AYZ bound for all $k$. Lincoln and Vyas~\cite{DBLP:conf/innovations/Lincoln020} showed that,
when $k$ is odd, this exponent
is optimal under some standard hypotheses. Dahlgaard et al.
introduced the so-called {\em $\Phi$-norm} of a vector and a matrix, which is then used to
bound the number of \emph{capped $k$-walks} in a \revision{$2k$-cycle-free} graph. Their technique is a
key building block for our results, and we will explain it in more detail in
Section~\ref{sec:background}.

On the listing side, Jin and Xu~\cite{DBLP:conf/stoc/JinX23} (and independently Abboud,
Khoury, Leibowitz, and Safier~\cite{DBLP:conf/fsttcs/AbboudKLS23}) showed how to list all
$4$-cycles in time $\tilde O(m^{4/3}+t)$, and Vassilevska~Williams and
Westover~\cite{DBLP:conf/innovations/WilliamsW25} devised an algorithm to list all
$6$-cycles in time $\tilde O(m^{8/5}+t)$. Both runtimes improve upon the corresponding AYZ
bounds of $\tilde O(m^{3/2})$ and $\tilde O(m^{5/3})$, respectively. They also proved a
conditional lower bound showing that $\tilde O(m^{4/3-\epsilon} + t)$ time is not possible
for listing $4$-cycles. In the dense-graph regime, an $\tilde O(n^2+t)$-time algorithm is
known only for listing $4$-cycles~\cite{DBLP:conf/stoc/JinX23} and
$6$-cycles~\cite{DBLP:conf/sosa/JinWZ24}. A natural open question is therefore to design an
algorithm that beats the AYZ bound for listing $2k$-cycles for all $k$.

One of the reasons researchers have focused on the even-cycle case is the existence of a
Tur\'{a}n-type theorem~\cite{Diestel2017} and supersaturation
results~\cite{ErdosSimonovits1983} for even cycles. In particular, a classic theorem of
Bondy and Simonovits~\cite{bondy1974cycles} states that any graph with \revision{$m \geq 100 k
n^{1+1/k}$} edges must contain a $2k$-cycle. \revision{(Note that, no such result exists for
odd cycles: the complete balanced bipartite graph $K_{n/2,n/2}$ has $\Theta(n^2)$ edges but no odd
cycle.)} Using this result, Dahlgaard et al.~\cite{DBLP:conf/stoc/DahlgaardKS17} showed that
a \revision{$2k$-cycle-free} graph must have a ``small'' number of \emph{capped $k$-walks}. We will
define this notion precisely in Section~\ref{sec:background}. At a high level, this
structural result bounds the bag size of a tree decomposition used to answer the $C_{2k}$
query, which is what enables algorithms faster than AYZ.
For the listing problem, however, a Tur\'{a}n-type theorem alone is not sufficient.
What we need is a \emph{quantitative} supersaturation result: a guarantee that a
sufficiently dense graph contains not just one $2k$-cycle, but \emph{many} --- enough
that we can charge the running time of our algorithm to the output size $t$.

Given that the exponent $2k/(k+1)$ is likely to be optimal for $C_{2k}$
detection~\cite{DBLP:conf/innovations/Lincoln020}, it is natural to ask whether the same
exponent works for $C_{2k}$-listing, i.e. whether we can list all $2k$-cycles in time
$\tilde O(m^{2k/(k+1)} + t)$. As shown in Vassilevska~Williams and
Westover~\cite{DBLP:conf/innovations/WilliamsW25} (attributed therein to Ce Jin and Renfei
Zhou), we can list $2k$-cycles in time $\tilde O(m^{2k/(k+1)} + t)$ if the following
supersaturation conjecture holds, which is stated as Conjecture~25
in~\cite{DBLP:conf/innovations/WilliamsW25}.

\begin{conj}
  Let $G = (A\cup B, E)$ be a bipartite graph with $m$ edges.
  If $m \geq \revision{100} k\bigl(|A| + |B| + (|A||B|)^{(k+1)/2k}\bigr)$, then the number $t$ of
  $2k$-cycles in $G$ satisfies
    $t \;\geq\; \Omega\!\left(\frac{m^{2k}}{|A|^k |B|^k}\right).$
    \label{conj:supersaturation}
\end{conj}

For $k=2$ the conjecture can be proved using elementary arguments, as we shall
illustrate in Section~\ref{sec:warm:ups}. For $k=3$, Vassilevska~Williams and
Westover~\cite{DBLP:conf/innovations/WilliamsW25} made some progress towards the conjecture,
thus obtaining an $\widetilde{O}(m^{8/5} + t)$-time algorithm for listing $6$-cycles in an
undirected graph. This was the state of the art in even-cycle listing prior to our work.

{\bf Our contributions.} Building on the work of Dahlgaard et
al.~\cite{DBLP:conf/stoc/DahlgaardKS17} and Vassilevska~Williams and
Westover~\cite{DBLP:conf/innovations/WilliamsW25}, we prove that all $2k$-cycles in an
undirected graph can be listed in time $O\!\left((m^{(2k^2-k+1)/(k^2+1)} + t)\log m\right)$ for all
$k \geq 2$. This is the first improvement to $2k$-cycle listing for all $k \geq 3$ since the
AYZ result 30 years ago. From a database perspective, our result indicates that queries with
self-joins and symmetric input relations can be evaluated more efficiently than what the
submodular width would predict, and that extremal and supersaturation results can be
powerful tools for query evaluation.

The technical contributions are as follows. Using a supersaturation result of Jiang and
Yepremyan~\cite{DBLP:journals/cpc/JiangY20}, we prove a weaker form of
Conjecture~\ref{conj:supersaturation}. We then exploit the capped-$k$-walk accounting
framework of Dahlgaard et al.~\cite{DBLP:conf/stoc/DahlgaardKS17} to bound the $\Phi$-norm
of the adjacency matrix of the graph, which in turn gives a tight bound on the number of
capped $k$-walks. The final algorithm is expressed via Boolean tensor decomposition ---
capturing the color-coding technique of AYZ --- where the query plan is a collection of tree
decompositions whose bag sizes are controlled by the capped-$k$-walk bound. Along the way,
our approach shaves a $\log^{k+1} n$ factor from the layered decomposition approach
of~\cite{DBLP:conf/innovations/WilliamsW25}.

For $k \geq 4$, our runtime strictly improves upon AYZ and is the first new result in this
regime since the original paper. For $k = 3$ (listing $6$-cycles), our exponent ($8/5$)
matches that of Vassilevska~Williams and Westover~\cite{DBLP:conf/innovations/WilliamsW25},
but our algorithm improves upon theirs in several respects: (1) we avoid the $O(\log^{74}
n)$ polylogarithmic overhead present in their algorithm; (2) we require no computer-assisted
case analysis (their proof involved approximately 576 cases); and (3) our query plan
is expressed entirely using \revision{explicit relational operators}, without BFS traversal, making it
naturally amenable to implementation in database systems.

A further benefit is that our algorithm and its analysis are simple: the algorithm can
be implemented via query rewriting, making it straightforward to deploy in practice. In
addition to the obvious application of our result to efficiently evaluating the $C_{2k}$
query in relational databases (and graph databases), we would also like to bring to the
attention of database-theory researchers the idea of proving and applying extremal and
supersaturation results to database query evaluation.

{\bf Outline.} The rest of the paper is organized as follows. In
Section~\ref{sec:background}, we introduce the necessary background and notations. In
particular, we explain how color-coding can be used to reduce a cycle-query to a closed-walk
query, and the $\Phi$-norm approach. In the ``warm-up'' Section~\ref{sec:warm:ups}, we
illustrate the supersaturation phenomenon and how it can be used to show that
Conjecture~\ref{conj:supersaturation} implies an algorithm to list $2k$-cycles in time
$\tilde O(m^{2k/(k+1)}+t)$. This result was already shown in Vassilevska~Williams and
Westover~\cite{DBLP:conf/innovations/WilliamsW25}, but our analysis is different, and more
importantly, the algorithm is a database query plan that does not require BFS traversal. The
section briefly presents a folklore result proving the conjecture for $C_4$, which implies
that $4$-cycles can be listed in time $\tilde O(m^{4/3}+t)$. Section~\ref{sec:results}
presents our main result for listing $2k$-cycles for all $k$. Finally, in
Section~\ref{sec:conclusions}, we conclude with some open questions and directions for
future work.

\section{Background}\label{sec:background}
\revision{
\subsection{Basic Definitions}
\label{subsec:basic:definitions}

Let $G = (V, E)$ be an undirected graph, where $V$ is the set of vertices and $E$ is the set of edges,
and let $X, Y \subseteq V$ be two subsets of vertices which can intersect.
We will use $E(X, Y)$ to denote the set of edges with one endpoint in $X$ and the other in $Y$.
If $X$ has only one vertex $x$, we will simply write $E(x, Y)$ instead of $E(\{x\}, Y)$.
A bipartite graph is a graph $G = (A \cup B, E)$ where the set of vertices is partitioned
into two disjoint sets $A$ and $B$, and every edge in $E$ has one endpoint in $A$ and the
other in $B$.

Let $\bm A=(a_{ij})$ be a square matrix. Its trace is defined as the sum of its diagonal elements, i.e.,
$\text{trace}(\bm A) = \sum_{i} a_{ii}$.
The Frobenius norm of $\bm A$, denoted by $\|\bm A\|_F$, is defined as
\[ \|\bm A\|_F = \sqrt{\sum_{i,j} a_{ij}^2}.  \]
It is easy to see that
\begin{align}
\|\bm A\|_F^2 &= \text{trace}(\bm A^\top \bm A). \label{eqn:trace}
\end{align}
}

\subsection{Tree decomposition and multiple query plans}
\label{subsec:td}

A {\em tree decomposition} (TD) of a hypergraph $\calH=(\calV,\calE)$ is a pair $(T, \chi)$,
where $T=(V(T),E(T))$ is a tree and $\chi: V(T) \to 2^{\calV}$ assigns to each node of $T$ a
set of vertices of $\calH$, called a \emph{bag}, subject to two conditions: (i) every
hyperedge $S \in \calE$ is contained in some bag $\chi(b)$; and (ii) for every vertex $v \in
\calV$, the set of nodes $\{b \in V(T) \mid v \in \chi(b)\}$ induces a connected subtree of
$T$.

Each conjunctive query (CQ) $Q(F) \text{ :- } \bigwedge_{S \in \calE} R_S(S)$ is naturally
associated with a hypergraph $\calH = (\calV, \calE)$, where $\calV$ is the set of variables
of $Q$, each hyperedge $S \in \calE$ is associated with an atom $R_S(S)$, and $F \subseteq
\calV$ is the set of free variables. In this paper we are only concerned with Boolean
conjunctive queries (i.e., $F = \emptyset$) and full conjunctive queries (i.e., $F =
\calV$). A tree decomposition of $Q$ is defined as a tree decomposition of its associated
hypergraph $\calH$. We abuse notation and use $S$ to refer to both the hyperedge in $\calE$
and the set of variables in the corresponding relation $R_S$.

There is essentially only one tree decomposition for both the $C_{2k}$-detection and
$C_{2k}$-listing queries~\eqref{eqn:ck:detection} and~\eqref{eqn:ck:listing}, consisting of
a single bag $\{v_1, v_2, \ldots, v_{2k}\}$ that contains all the variables. This is due to
the $\neq$ atoms that connect all the variables together into a clique. On the other hand,
the $W_{2k}$-detection and $W_{2k}$-listing queries~\eqref{eqn:wk:detection}
and~\eqref{eqn:wk:listing} have many non-trivial tree decompositions, as shown in
Figure~\ref{fig:td}. In the figure, we omit the tree decompositions themselves and show only
the structures of the bags. On the left, every triangle drawn is a bag and two bags are
adjacent if they share an edge. On the right, we have a tree decomposition with two bags,
one containing all variables on the top half of the cycle and the other containing all
variables on the bottom half of the cycle.
\begin{figure}[th]
\centering
\begin{tikzpicture}[
    every node/.style={font=\small},
    vertex/.style={circle, draw, fill=white, inner sep=1.2pt},
]

\begin{scope}[shift={(0,0)}]
    \foreach \k in {0,...,9} {
        \pgfmathsetmacro{\ang}{90 - \k*36}
        \coordinate (L\k) at (\ang:2);
    }

    \foreach \k in {0,...,9} {
        \pgfmathtruncatemacro{\next}{mod(\k+1,10)}
        \ifnum\k<4
            \draw (L\k) -- (L\next);
        \else
            \ifnum\k>8
                \draw (L\k) -- (L\next);
            \else
                \draw[dashed] (L\k) -- (L\next);
            \fi
        \fi
    }

    \foreach \k in {2,3,4,5,6,7,8} {
        \draw[dashed] (L0) -- (L\k);
    }

    \node[vertex] at (L0) {};
    \node[vertex] at (L8) {};  

    \node[above]      at (L0) {$v_i$};
    \node[above left] at (L9) {$v_{i-1}$};
    \node[left]       at (L8) {$v_{i-2}$};
    \node[above right]at (L1) {$v_{i+1}$};
    \node[right]      at (L2) {$v_{i+2}$};
\end{scope}

\begin{scope}[shift={(7,0)}]
    \foreach \k in {0,...,9} {
        \pgfmathsetmacro{\ang}{180 - \k*36}
        \coordinate (R\k) at (\ang:2);
    }

    \foreach \k in {0,...,9} {
        \pgfmathtruncatemacro{\next}{mod(\k+1,10)}
        \ifnum\k<2
            \draw (R\k) -- (R\next);
        \else
            \ifnum\k>7
                \draw (R\k) -- (R\next);
            \else
                \draw[dashed] (R\k) -- (R\next);
            \fi
        \fi
    }

    \draw[dashed] (R0) -- (R5);

    \node[vertex] at (R0) {};   
    \node[vertex] at (R1) {};   
    \node[vertex] at (R9) {};   
    \node[vertex] at (R5) {};   

    \node[left]       at (R0) {$v_i$};
    \node[above]      at (R1) {$v_{i+1}$};
    \node[below]      at (R9) {$v_{i-1}$};
    \node[right]      at (R5) {$v_{i+k-1}$};
\end{scope}
\end{tikzpicture}
\caption{Tree decompositions for the closed-walk queries $W_{2k}$}
\label{fig:td}
\end{figure}
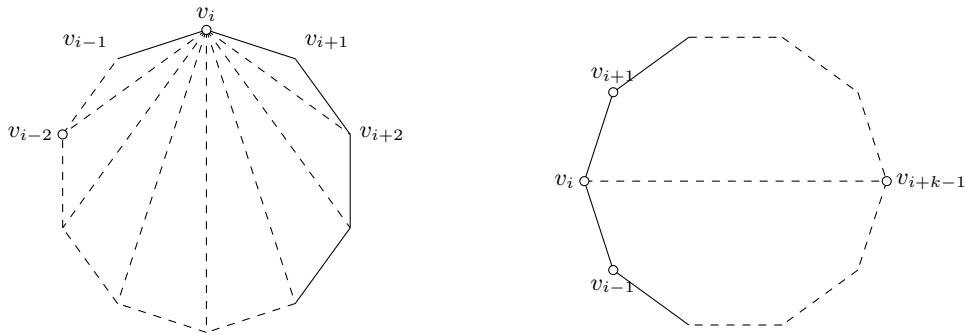

A tree decomposition of a query represents a query plan~\cite{DBLP:conf/vldb/Yannakakis81,
DBLP:journals/talg/GroheM14, DBLP:conf/pods/KhamisNR16}. Specifically, given a tree
decomposition $(T, \chi)$ of $W_{2k}$, for each node $b \in V(T)$ we evaluate the
restriction\footnote{We project all input relations that share variables with the bag onto
the variables in the bag, and evaluate the full conjunctive query on these projected
relations.} of $W_{2k}$ to the variables in the bag $\chi(b)$, yielding an intermediate
relation $I_{\chi(b)}$ over the variables $\chi(b)$. The full query result can then be
obtained by joining all the intermediate relations $\{I_{\chi(b)}\}_{b \in V(T)}$ along the
tree $T$, following Yannakakis's algorithm~\cite{DBLP:conf/vldb/Yannakakis81}. The runtime
of this procedure is bounded by the \revision{worst-case size-bound of the largest
intermediate relation, subject to input cardinality or degree
constraints~\cite{theoretics:13722}}. In particular, if the query is a Boolean CQ, then the
runtime is $O(N+\max_{b \in V(T)} \norm{I_{\chi(b)}})$, and if the query is a full CQ (i.e.,
a listing query) then the runtime is $O(N+\max_{b \in V(T)} \norm{I_{\chi(b)}} + t)$, where
$N$ is the input size, $t$ is the output size, \revision{and $\norm{I}$ denote the
worst-case size bound of the intermediate relation $I$ subject to input cardinality or
degree constraints}.

One of the fundamental ideas behind both the AYZ
algorithm~\cite{AYZ97} and Marx's submodular
width~\cite{DBLP:journals/jacm/Marx13} is that one can use \emph{multiple} query plans for a
single query, by partitioning the data appropriately and handling each partition with a
different tree decomposition.

To illustrate, consider the $W_{2k}$-detection and listing queries. Let $\Delta$ be a
parameter to be fixed later. Partition the vertex set into \emph{heavy} vertices (those of
degree $> \Delta$) and \emph{light} vertices (all others), denoted by $V := H \cup L$. AYZ's
query plan rewrites the $W_{2k}$ query as a disjunction of $2k+1$ sub-queries, each of which
can be answered by a different tree decomposition. In particular, for each $i \in [2k]$, we
have a sub-query $W^{h,i}_{2k}$ that looks for a closed walk of length $2k$ where the vertex
$v_i$ is heavy. Additionally, there is one more sub-query $W^\ell_{2k}$ that looks for a
closed walk of length $2k$ where all vertices are light:
\begin{align}
    W^{h,i}_{2k}(v_1,\dots,v_{2k}) &\leftarrow E(v_1,v_2) \wedge E(v_2,v_3) \wedge \cdots \wedge E(v_{2k-1},v_{2k})
    \wedge E(v_{2k},v_1) \wedge \underbrace{H(v_i)}_{\mathrm{heavy~vertex}}
    \\
    W^\ell_{2k}(v_1,\dots,v_{2k}) &\leftarrow E(v_1,v_2) \wedge E(v_2,v_3) \wedge \cdots
    \wedge E(v_{2k-1},v_{2k}) \wedge E(v_{2k},v_1)
    \wedge \underbrace{\bigwedge_{j \in [2k]} L(v_j)}_{\mathrm{all~light~vertices}}
\end{align}
The first $2k$ tree decompositions, shown on the left of Figure~\ref{fig:td}, handle the
$W^{h,i}_{2k}$ sub-queries. The remaining tree decomposition, shown on the right of
Figure~\ref{fig:td}, handles the all-light case with two bags, one covering the top half and
the other covering the bottom half of the walk.

For the heavy query plans, since there are at most $O(m/\Delta)$ heavy vertices, every
intermediate relation corresponding to a bag of such a plan has size at most
$O(m^2/\Delta)$. For the all-light query plan, each of the two bags is bounded in size by
$O(m \Delta^{k-1})$, since each light vertex has degree bounded by $\Delta$. Choosing
$\Delta = m^{1/k}$ to balance the two bounds gives $O(m^2/\Delta) = O(m \Delta^{k-1}) =
O(m^{2-1/k})$, and AYZ achieves $O(m^{2-1/k})$ time for $W_{2k}$-detection and $O(m^{2-1/k}
+ t)$ time for $W_{2k}$-listing.

\subsection{Color-coding and Boolean tensor decomposition}
\label{subsec:color:coding}

As noted above, the $C_{2k}$ queries have only one
(trivial) tree decomposition --- a single bag containing all variables --- making a direct
application of the query-plan framework prohibitively expensive. The \emph{color-coding}
technique~\cite{color-coding} overcomes this obstacle by effectively reducing the cycle
query to a closed-walk query. This reduction can be done with query rewriting, using Boolean
tensor decomposition~\cite{DBLP:conf/icdt/KhamisNOS19} as follows.

The predicates that distinguish the $C_{2k}$ query from the $W_{2k}$ query are exactly the
distinctness constraints: $\bigwedge_{i \neq j} V(v_i) \wedge V(v_j) \wedge v_i \neq v_j,$
where $V$ denotes the vertex set of the graph. Using color-coding, one can compute in linear
time $r \cdot 2k$ unary relations $C_{\ell,i}$, for $\ell \in [r]$ and $i \in [2k]$, with $r
= 2^{O(k)} \log n$, such that
\begin{align}
  \bigwedge_{i \neq j} V(v_i) \wedge V(v_j) \wedge v_i \neq v_j
  \;\equiv\;
  \bigvee_{\ell=1}^{r} \bigwedge_{i \in [2k]} C_{\ell,i}(v_i).
  \label{eqn:tensor:decomp}
\end{align}
We do not need to materialize these unary relations $C_{\ell,i}$, as they are defined by
hash functions.
\revision{The number $r$ is the ``tensor rank'' of the Boolean tensor defined by the left
hand side of~\eqref{eqn:tensor:decomp}. This rank is computed from the sizes of certain
family of hash functions used to derandomize the color-coding procedure~\cite{DBLP:conf/icdt/KhamisNOS19}.
We skipped all this machinery and stated only the final result that is needed for this paper.
We refer the reader to~\cite{DBLP:conf/icdt/KhamisNOS19} for the full details.}

Substituting~\eqref{eqn:tensor:decomp} into the $C_{2k}$ query and distributing the disjunction reduces the
cycle query to $r$ instances of the closed-walk query $W_{2k}$, each augmented with unary
filters $C_{\ell,i}$ on the variables:
\begin{align}
    E(v_1,v_2) \wedge \cdots \wedge E(v_{2k},v_1) \wedge \underbrace{\bigwedge_{i \neq j} v_i \neq v_j}_{\mathrm{distinctness}}
    \equiv\;& E(v_1,v_2) \wedge \cdots \wedge E(v_{2k},v_1) \wedge
        \underbrace{\bigvee_{\ell \in [r]} \bigwedge_{i \in [2k]} C_{\ell,i}(v_i)}_{\text{every $v_i$ has a ``color''}} \\
    \equiv\;& \bigvee_{\ell \in [r]} \Bigl[E(v_1,v_2) \wedge \cdots \wedge E(v_{2k},v_1) \wedge
        \underbrace{\bigwedge_{i \in [2k]} C_{\ell,i}(v_i)}_{\mathrm{filtering~redistribution}}\Bigr].
\end{align}

Each such instance can then be answered using AYZ's multi-TD query plan for closed walks, at
a cost of $O(m^{2-1/k})$ per instance. The total runtime for $C_{2k}$-detection is therefore
$\tilde{O}(m^{2-1/k})$, and for $C_{2k}$-listing it is $\tilde{O}(m^{2-1/k} + t)$,
recovering the AYZ bounds for even cycles. Here, $\tilde{O}$ hides a $2^{O(k)} \log n$
factor.

\begin{rmk}
Here $t$ counts \emph{cycles}, not closed walks, so the listing algorithm is
output-sensitive in the number of cycles. One cannot simply drop the distinctness
constraints, enumerate all closed walks, and filter out the non-cycles: the number of closed
walks can be far larger than the number of cycles, so this naive approach pays for walks
rather than for cycles. A tree, for example, has no cycles but plenty of even-length
closed walks.
\end{rmk}

\subsection{Capped $k$-walks and the $\Phi$-norm}
\label{subsec:phi:norm}

This section recasts the key ideas of Dahlgaard, Knudsen, and
St{\"{o}}ckel~\cite{DBLP:conf/stoc/DahlgaardKS17} in the language of database queries and
tree decompositions. Beyond being more familiar to database-theory readers, this view is
also more directly implementable: it is not obvious how to realize the algorithms
of~\cite{DBLP:conf/stoc/DahlgaardKS17} as typical database query plans.

The AYZ analysis in the previous section uses only a single tree decomposition for the
all-light case (the right side of Figure~\ref{fig:td}). Dahlgaard, Knudsen, and
St{\"{o}}ckel~\cite{DBLP:conf/stoc/DahlgaardKS17} made a sharp observation: because the
graph is undirected, one can break an even cycle into two halves at the vertex $v_i$ of
highest degree on the cycle. This is analogous to the $2k$ heavy TDs on the left of
Figure~\ref{fig:td}, where \revision{there is one tree decomposition for each posible
position $i \in [2k]$ on the cycle where a heavy vertex occurs. The same strategy can be
applied to the tree decomposition on the right side as well, where there is one tree
decomposition for each position $i \in [2k]$ on the cycle where the highest-degree vertex
occurs. This strategy gives us more information to bound the worst-case bag-size
of the bags in the tree decompositions on the right side.}

Specifically, given a graph $G = (V, E)$ with maximum degree $\Delta$, let $V_p$ denote the set of vertices whose degree
lies in the interval $(2^{p-1}, 2^p]$, for $p = 0, 1, \ldots, \lceil \log \Delta \rceil$.
Let $\bar{V}_p := \bigcup_{q \leq p} V_q$ be the set of vertices of degree at most $2^p$. We
proceed as in the AYZ algorithm above, but handle the all-light data partition differently:
we rewrite the all-light query as a disjunction of $2k \cdot \lceil \log \Delta \rceil$
sub-queries:
\begin{align}
    &E(v_1,v_2) \wedge \cdots \wedge E(v_{2k},v_1) \wedge \bigwedge_{i \in [2k]} L(v_i) \\
    \equiv\;& E(v_1,v_2) \wedge \cdots \wedge E(v_{2k},v_1) \wedge \bigwedge_{i \in [2k]} L(v_i)  \wedge
        \underbrace{\bigvee_{\substack{i \in [2k] \\ p \in [\lceil\log\Delta\rceil]}}\Bigl[V_p(v_i) \wedge
        \bigwedge_{j \neq i} \bar{V}_p(v_j)\Bigr]}_{\text{for some $i$, $v_i$ is of highest degree}} \\
    \equiv\;& \bigvee_{\substack{i \in [2k] \\ p \in [\lceil\log\Delta\rceil]}} \Bigl[E(v_1,v_2) \wedge \cdots \wedge E(v_{2k},v_1)
    \wedge \bigwedge_{i \in [2k]} L(v_i)  \wedge
        \underbrace{V_p(v_i) \wedge \bigwedge_{j \neq i} \bar{V}_p(v_j)}_{v_i \text{ in the highest degree bucket}}\Bigr].
\end{align}
We will use the TD on the right of Figure~\ref{fig:td} to handle each of these sub-queries,
where the vertices $v_i$ and $v_{i+k-1}$ (circularly) are anchored in the middle of the
cycle and the two bags are split at $v_i$ and $v_{i+k-1}$. In graph-theoretic terms, each
sub-query looks for cycles that pass through a vertex in $V_p$ with all other vertices in
$\bar{V}_p$. It is clear that the union of the answers to these sub-queries equals the
answer to the original query. The main task is then to bound the size of the intermediate
relations produced by each of these tree decompositions.

For each such ``light'' TD anchored at position $i$ and level $p$, the bag size is bounded
by the number of $k$-walks in $G$ that start from a vertex in $V_p$ and visit only vertices
in $\bar{V}_p$, \revision{for $1 \leq p \leq \lceil \log \Delta\rceil $}. 
We call these \emph{capped $k$-walks}, and denote their total count by
$\kappa_k(G, \Delta)$.

To bound $\kappa_k(G)$, Dahlgaard et al.~\cite{DBLP:conf/stoc/DahlgaardKS17} introduced a
powerful analytical instrument: the \emph{$\Phi$-norm}. \revision{ The quantity arises
naturally from the proof of Lemma~\ref{lmm:capped-walks} below, whose proof can be found in
the appendix. In particular, the reader can find some more intuition about the $\Phi$-norm
and its role in bounding capped $k$-walks in the appendix. }

\begin{defn}[$\Phi$-norm]
The \emph{$\Phi$-norm} of a vector $\bm{v}$ is defined as
\begin{align}
  \phinorm{\bm{v}} = \int_0^\infty \sqrt{|\{ i : |v_i| \geq x \}|} \, dx.
\end{align}

The \emph{$\Phi$-norm} of a matrix $\bm{A}$ is defined as the operator norm induced by
the vector $\Phi$-norm:
\begin{align}
  \phinorm{\bm{A}} = \sup_{\bm{v} \neq \bm{0}} \frac{\phinorm{\bm{A}\bm{v}}}{\phinorm{\bm{v}}}.
  \label{eqn:matrix:phi:norm}
\end{align}
\end{defn}

Throughout the paper, we write $A \lesssim B$ to mean $A = O(B)$.
The following bound on \(\kappa_k(G)\) via the \(\Phi\)-norm can be inferred from the proof
of Lemma~1.7 of~\cite{DBLP:conf/stoc/DahlgaardKS17}.
\revision{We include a proof of the lemma in the appendix.}

\begin{lmm}[\cite{DBLP:conf/stoc/DahlgaardKS17}]
\label{lmm:capped-walks}
\revision{Let $G$ be an undirected graph, and $\Delta$ a positive integer.}
Let $\bm{A}_p$ denote the adjacency matrix of the subgraph of $G$, induced on
vertices in $\bar{V}_p$, \revision{where $p \in [\lceil \log \Delta \rceil ]$}
Then,
\begin{align}
  \kappa_k(G, \Delta) \lesssim m \cdot (\max_{p} \, \phinorm{\bm{A}_p})^{k-1}.
\end{align}
\end{lmm}

The $\Phi$-norm of a matrix can be bounded via the following ``indicator'' version:

\begin{lmm}[\revision{Lemma 5 of \cite{DBLP:conf/stoc/DahlgaardKS17}}]
\label{lmm:indicator-phi-norm}
Let $\bm A$ be a real $n \times n$ matrix. If, for all vectors $\bm v \in \{0,1\}^n$ we have
$ \phinorm{\bm{A}\bm{v}} \leq C  \phinorm{\bm{v}}$ for some value $C$, then $
\phinorm{\bm{A}} \leq 16C$.
\end{lmm}

\revision{In particular, instead of checking the $\Phi$-norm condition for all vectors, it suffices to
check it for all indicator vectors $\bm v \in \{0,1\}^n$.}

\section{Warm-ups}\label{sec:warm:ups}
This section presents three warmup results that illustrate how Tur\'an-type bounds and the
more general supersaturation phenomenon enable efficient algorithms for $2k$-closed walks
and $2k$-cycles on undirected graphs. We begin with a simple supersaturation result that
lets us solve $W_{2k}$-detection in $O(1)$ time and $W_{2k}$-listing in output-linear time.
Then, we prove that Conjecture~\ref{conj:supersaturation} implies a
$\tilde O(m^{2k/(k+1)}+t)$-time algorithm for $2k$-cycle listing, matching the
$C_{2k}$-detection exponent of Dahlgaard, Knudsen, and
St\"ockel~\cite{DBLP:conf/stoc/DahlgaardKS17}. This result was already established by
Vassilevska Williams and Westover~\cite{DBLP:conf/innovations/WilliamsW25} (attributed
to Jin and Zhou); however, our analysis differs from theirs, and our algorithm --- which
avoids BFS traversal --- is friendlier to database-system implementations. Finally, as an
application of this result, we show that $C_4$-listing can be done in $\tilde O(m^{4/3}+t)$
time by proving that the supersaturation conjecture holds for $4$-cycles. This fact is
already known and is folklore.

\subsection{The $2k$-closed walk queries}

\revision{
This section proves the following simple proposition for two reasons. The main reason is
that it is an extremely light-weight introduction to the idea of using a supersaturation
bound in analyzing a query evaluation algorithm. The secondary reason is that it shows very
nicely the distinction between walk queries and cycle queries on undirected graphs. Note
that there is no such result in the directed case, as far as we know.
}

\begin{prop}
    Given an undirected graph $G=(V,E)$ and an integer $k \geq 2$, we can answer the
    Boolean $W_{2k}$-detection query in $O(1)$ time and the $W_{2k}$-listing query in
    $O(1+t)$ time, where $t$ is the output size.
\end{prop}
\begin{proof}
    Since we can always walk back and forth on an edge, the $W_{2k}$-detection query is
    trivially true if $G$ has at least one edge. This is a trivial example of a
    Tur\'an-type theorem for $2k$-closed walks. Note that the reasoning does not work
    for odd-length closed walks.

    Next, consider the listing version. Using the tree decomposition from the right of
    Figure~\ref{fig:td}, we can list $2k$-closed walks in time $O(\text{\# $k$-walks} + t)$,
    as explained in Section~\ref{subsec:td}. Thus, it is sufficient to prove that
    $\text{\# $k$-walks} \leq t$.

    Let $\bm A$ be the adjacency matrix of $G$ and $\bm 1$ be the all-$1$ vector of
    dimension $n$, where $n$ is the number of vertices in $G$. The number of $k$-walks in
    $G$ is $\bm 1^\top \bm A^k \bm 1$, and the number of $2k$-closed walks in $G$ is
    \revision{(recall some facts from Section~\ref{subsec:basic:definitions}, in particular
    identity~\eqref{eqn:trace})}
    \begin{align*}
        t = \text{trace}(\bm A^{2k}) &= \text{trace}((\bm A^k)^\top \bm A^k)
        = \norm{\bm A^k}_F^2
        = \sum_{u,v} (\bm A^k)_{u,v}^2
        \geq \sum_{u,v} (\bm A^k)_{u,v}
        = \bm 1^\top \bm A^k \bm 1
        = \text{\# $k$-walks}.
    \end{align*}
\end{proof}

\subsection{Why the supersaturation conjecture implies efficient $2k$-cycle listing}

\revision{This section proves that Conjecture~\ref{conj:supersaturation} implies efficient
$2k$-cycle listing. The idea is to use the conjecture to bound the $\Phi$-norm of the
adjacency matrix in Lemma~\ref{lmm:conj:supersaturation:implies:Phi:norm} below. This lemma
along with Lemma~\ref{lmm:capped-walks} allows us to bound the number of capped $k$-walks.
Lastly, in Proposition~\ref{prop:2k:cycle:listing} applies this bound to show that we can
adapt AYZ's algorithm with a different heavy/light threshold to obtain an efficient
$2k$-cycle listing algorithm.}

\revision{The novelty of this section is in the proof of
Lemma~\ref{lmm:conj:supersaturation:implies:Phi:norm}.} Our analysis differs from those
in~\cite{DBLP:conf/innovations/WilliamsW25,DBLP:conf/stoc/DahlgaardKS17}, and so may be of
independent interest. \revision{ In particular, our analysis is shorter, and we made very
transparent what the constraints are. The bound is just an optimization problem (over
integrals) over these constraints. In contrast,
both~\cite{DBLP:conf/innovations/WilliamsW25,DBLP:conf/stoc/DahlgaardKS17} use combinatorial
arguments which are more ``magical". Especially compared
to~\cite{DBLP:conf/innovations/WilliamsW25}, we bound the Phi-norm directly by a case
analysis on which constraint is tight, instead of a layered decomposition argument, which is
also what removes the $log^{k+1} n$ factor.}

\begin{lmm}
Suppose Conjecture~\ref{conj:supersaturation} holds. Let $\bm A$ be the adjacency matrix
of a graph with maximum degree $\Delta$, $m$ edges, and $t$ $2k$-cycles.
Then, $\norm{\bm A}_\Phi = O(m^{1/(k+1)} + t^{1/(2k)}\log\Delta + \sqrt\Delta)$.
\label{lmm:conj:supersaturation:implies:Phi:norm}
\end{lmm}
\begin{proof}
From Lemma~\ref{lmm:indicator-phi-norm}, it is sufficient to show that, for any subset
$S$ of vertices, we have $\revision{\norm{\bm A \bm 1_S}_\Phi} = O(m^{1/(k+1)} + t^{1/(2k)}\log \Delta + \sqrt\Delta) \sqrt{|S|}$.
Let $\bm 1_S$ denote the indicator vector of $S$, let $E(v,S)$ denote the set of edges
between $v$ and $S$, and let $E(B, S)$ denote the set of edges between the set of vertices
$B$ and $S$. Define
\begin{align}
    B(x) &:= \{ v : |E(v,S)| \geq x \}, && f(x) := |B(x)|, && s := |S|
\end{align}
Then, our goal is to bound
$\phinorm{\bm A\bm 1_S} = \int_{x=0}^{\Delta} \sqrt{f(x)} \, dx$,
subject to the constraints we know about $f$.
Since there are at most $m$ edges in total, and since every vertex in $S$ has degree
at most $\Delta$, we have
\begin{align}
    x \cdot f(x) &\leq 2m,
    && x \cdot f(x) \leq s\Delta,
    && \forall x \leq \Delta.
    \label{eq:edge:count:constraint}
\end{align}

Next, from Conjecture~\ref{conj:supersaturation}
\revision{applied to the bipartite graph $(B(x) \cup S, E(B(x), S))$, one of the following
inequalities must hold:
\begin{align*}
\text{Either } && |E(B(x), S)| &< 100 k \bigl( f(x) + s + (f(x)s)^{(k+1)/(2k)}\bigr) \\
\text{or } && t &= \Omega\left( \frac{|E(B(x), S)|^{2k}}{f(x)^k s^k} \right)
\end{align*}
In other words,
\begin{align*}
\text{Either } && |E(B(x), S)| &< 100 k \bigl( f(x) + s + (f(x)s)^{(k+1)/(2k)}\bigr) \\
\text{or } && |E(B(x), S)| &= O\bigl( t^{1/(2k)} \sqrt{f(x)s}\bigr)
\end{align*}
Consequently,
}
there is a constant $C$ (depending on $k$) such that
\begin{align*}
    x \cdot f(x) \leq |E(B(x),S)| \leq
    C(f(x) + s + (f(x)s)^{(k+1)/2k} + t^{1/2k} \sqrt{f(x)s}).
\end{align*}
Thus, when $x \geq 2C$ the following holds:
\begin{align*}
    \frac x 2 \cdot f(x)
    &\leq C(s + (f(x)s)^{(k+1)/2k} + t^{1/2k} \sqrt{f(x)s}) \\
    &\leq 3C\max\bigl\{s, (f(x)s)^{(k+1)/2k}, t^{1/2k} \sqrt{f(x)s}\bigr\}.
\end{align*}
As a consequence, when $x \geq 2C$, one of the following three cases must hold:
\begin{itemize}
\item $x \cdot f(x) \leq 6Cs$, implying $\sqrt{f(x)} \leq O(\sqrt{s/x})$.
\item $x \cdot f(x) \leq 6C(f(x)s)^{(k+1)/2k}$, implying $\sqrt{f(x)} \leq O(s^{(k+1)/(2(k-1))} / x^{k/(k-1)})$.
\item $x \cdot f(x) \leq 6C t^{1/2k} \sqrt{f(x)s}$, implying $\sqrt{f(x)} \leq O(t^{1/(2k)} \sqrt{s} / x)$.
\end{itemize}
Combining these with the trivial bound $\sqrt{f(x)} = O(\sqrt{m/x})$, we have (recall
that $A \lesssim B$ means $A = O(B)$)
\begin{align*}
    \sqrt{f(x)} \lesssim
        \sqrt{\frac s x} +
        \min\left\{ \sqrt{\frac m x}, \frac{s^{(k+1)/(2(k-1))}}{x^{k/(k-1)}} \right\} +
        \frac{t^{1/(2k)} \sqrt{s}}{x}.
\end{align*}
Thus, we can bound $\phinorm{\bm A\bm 1_S}$ as follows, using $\sqrt{f(x)} \leq \sqrt{s\Delta}$ for $x \leq 2C$:
\begin{align*}
    \phinorm{\bm A\bm 1_S}
    &\lesssim \int_{x=0}^{2C} \sqrt{f(x)} \, dx + \int_{x=2C}^{\Delta} \sqrt{f(x)} \, dx \\
    &\lesssim \int_{0}^{2C} \sqrt{s\Delta} \, dx
    + \int_{x=2C}^{\Delta} \sqrt{\frac s x} dx
    + \int_{x=2C}^{\Delta} \min\left\{ \sqrt{\frac m x}, \frac{s^{(k+1)/(2(k-1))}}{x^{k/(k-1)}} \right\} dx
    + \int_{x=2C}^{\Delta} \frac{t^{1/(2k)} \sqrt{s}}{x} \, dx \\
    &\lesssim \sqrt{s\Delta} + t^{1/(2k)} \sqrt{s} \log \Delta +
        \int_{x=2C}^{\Delta} \min\left\{ \sqrt{\frac m x}, \frac{s^{(k+1)/(2(k-1))}}{x^{k/(k-1)}} \right\} dx.
\end{align*}
Since
\begin{align*}
 \sqrt{\frac{m}{x}} \leq \frac{s^{(k+1)/(2(k-1))}}{x^{k/(k-1)}}
    &\iff x \leq \frac{s}{m^{(k-1)/(k+1)}},
\end{align*}
setting $x^* = s/m^{(k-1)/(k+1)}$, we break the last integral into two parts:
\begin{align*}
    \int_{x=2C}^{\Delta} \min\left\{ \sqrt{\frac m x}, \frac{s^{(k+1)/(2(k-1))}}{x^{k/(k-1)}} \right\} dx
    \leq \int_{x=2C}^{x^*} \sqrt{\frac m x} dx
    + \int_{x=x^*}^{\Delta} \frac{s^{(k+1)/(2(k-1))}}{x^{k/(k-1)}} dx
    \lesssim m^{1/(k+1)} s^{1/2}.
\end{align*}
Overall, we have $\phinorm{\bm A\bm 1_S} \lesssim
(\sqrt{\Delta} + t^{1/(2k)} \log \Delta + m^{1/(k+1)}) \sqrt s$ as desired.
\end{proof}

\begin{prop}
    If Conjecture~\ref{conj:supersaturation} holds, then we can list $t$ $2k$-cycles in time
    $\tilde O(m^{2k/(k+1)} + t)$.
    \label{prop:2k:cycle:listing}
\end{prop}
\begin{proof}
From the AYZ approach described in Section~\ref{subsec:td} and Section~\ref{subsec:color:coding},
call a vertex ``heavy'' if its degree is more than
$\Delta$. The number of heavy vertices is at most $2m/\Delta$; hence, the heavy cycles
can be listed in time $\tilde O(m^2/\Delta + t)$.
The cycles involving light vertices can be listed in time $\tilde O(\kappa_k(G) + t)$. From
Lemma~\ref{lmm:capped-walks} and the bound on $\norm{\bm A}_\Phi$ from
Lemma~\ref{lmm:conj:supersaturation:implies:Phi:norm}, we have
\begin{align*}
\kappa_k(G)
    &\lesssim m \cdot \max_{p} \phinorm{\bm A_p}^{k-1}\\
    &\lesssim m \cdot \left(
        \sqrt{\Delta} + t^{1/(2k)} \log \Delta + m^{1/(k+1)}
   \right)^{k-1} \\
    &\lesssim m \cdot
        \Delta^{\frac{k-1}{2}} + m^{\frac{k-1}{k+1}} + t^{\frac{k-1}{2k}} \log^{k-1}\Delta\\
    &=
        m\Delta^{\frac{k-1}{2}} + m^{\frac{2k}{k+1}} + mt^{\frac{k-1}{2k}} \log^{k-1}\Delta\\
    (\text{Young's inequality})
    &\lesssim
        m\Delta^{\frac{k-1}{2}} + m^{\frac{2k}{k+1}} + m^{\frac{2k}{k+1}} (\log\Delta)^{\frac{2k(k-1)}{k+1}} + t.
\end{align*}
We balance the terms by setting $\Delta = m^{2/(k+1)}$, yielding the desired runtime.
\end{proof}

\subsection{The $4$-cycle query}

As an application of Lemma~\ref{lmm:conj:supersaturation:implies:Phi:norm}, we show
how to obtain the $\tilde O(m^{4/3} + t)$-time algorithm for $C_4$-listing by reproving
the following unbalanced supersaturation result for $4$-cycles, a special case of
Conjecture~\ref{conj:supersaturation} that is already known in the literature. We include
the proof here as a warm-up.

\begin{prop}[Unbalanced supersaturation for \(4\)-cycles: the \(k=2\) case]
Given a bipartite graph \(G=(A\cup B,E)\) with \(a:=|A|\), \(b:=|B|\) and
\revision{\(m \geq 4(a+b+(ab)^{3/4})\)} edges, then there are
\(\Omega(m^4/(a^2b^2))\) copies of \revision{$4$-cycles}.
\end{prop}

\begin{proof}
Let \(t\) be the number of \revision{$4$-cycles} and assume that \(a\ge b\) without loss of generality.
For \(u\in A\), let \(d(u)\) be its degree. The number of \(2\)-paths whose
middle vertex lies in \(A\) is
\begin{align}
    P:=\sum_{u\in A}\binom{d(u)}{2}
    \ge a\binom{\frac1a\sum_{u\in A}d(u)}{2}
    =a\binom{m/a}{2}
    \revision{\geq a\frac{(m/a)^2}{4}}
    = \frac{m^2}{4a},
    \label{eqn:P:lower:bound}
\end{align}
\revision{The first inequality follows from the convexity of the function $x\mapsto \binom{x}{2}$
and Jensen's inequality. The second inequality follows from the fact that
$\binom{x}{2}\ge \frac{x^2}{4}$ for $x\ge 2$, and that $m/a\ge 2$.}

For \revision{a pair of distinct vertices} \(x,y\in B\), let \(q(x,y):=|N(x)\cap N(y)|\)
\revision{be the number of common neighbors}. Then \revision{the number of $2$-paths whose
middle vertex lies in \(A\) is}
\(P=\sum_{\{x,y\}\subseteq B, x \neq y}q(x,y)\), and \revision{the number of $4$-cycles is}
\[
    t=\sum_{\{x,y\}\subseteq B, x \neq y}\binom{q(x,y)}{2}
    \ge \binom b2\binom{P/\binom b2}{2}
    \ge \frac{P^2}{4\binom b2}
    \ge \frac{P^2}{2b^2}
    \ge \frac{m^4}{32a^2b^2}.
\]
\revision{Here, similar to~\eqref{eqn:P:lower:bound}, we applied Jensen's inequality
and the convexity of the function $x\mapsto \binom{x}{2}$ to get the first inequality.
The second inequality follows from the fact that $\binom{x}{2}\ge \frac{x^2}{4}$ for $x\ge 2$.
and that \(P/\binom b2\ge 2\)}. This last fact follows from $a \geq b$ and the assumption on \(m\):
\(m\ge 4(ab)^{3/4}\ge 4a^{1/2}b\), which implies
\(P\ge \frac{m^2}{4a}\ge b^2\ge 2\binom b2\).
\end{proof}

\begin{cor}
    We can list $t$ $4$-cycles in time $ O(m^{4/3}\log n  + t \log n )$.
\end{cor}

\section{Main Results}\label{sec:results}
The key observation is that the structure of the proof of
Theorem~1.1 of~\cite{DBLP:conf/stoc/DahlgaardKS17} can be adapted
to the listing setting. The first step is to replace the
edge-bounding Lemma~3.4 of~\cite{DBLP:conf/stoc/DahlgaardKS17},
which assumes $2k$-cycles-freeness, with a supersaturation version
that depends on the number \(t\) of cycles of length $2k$.
We will use the following supersaturation result of Jiang and
Yepremyan~\cite{DBLP:journals/cpc/JiangY20}.

\begin{lmm}[Jiang and Yepremyan~\cite{DBLP:journals/cpc/JiangY20}]
    Given a positive integer constant $k$, there are
    positive constants $c$ and $d$,
    both depend on $k$, such that the following holds.
    Let $G$ be a graph with $n$ vertices and $m \geq c n^{1+1/k}$ edges. Then, $G$ contains
    at least $d(m/n)^{2k}$ cycles of length $2k$.
    \label{lmm:jiang:yepremyan}
\end{lmm}

We can use the aforementioned result can be used to obtain the following proposition. In what
follows,
the $\lesssim, \gtrsim$ symbols, abbreviating $\Omega$ and $O$, hide (somewhat cumbersome)
functions in $k$. It is not hard to keep track on the dependence on $c,d$ and the
exponential in $k$ factors arising throughout; we hide them for clarity of exposition.
\revision{The idea to randomly permute $A$, split it into $r = |A|/|B|$ parts of size $|B|$
(except potentially the last one), apply Jiang–Yepremyan's result inside each part to estimate
the expected number of cycles from this permutation, and then scale back appropriately to get
a lower bound on the total number of cycles}.

\begin{prop}
Let $G=(A\cup B, E)$ be a bipartite graph with $t$ cycles of length $2k,$ $a=|A|\geq b=|B|$. Then,
$m=|E| \revision{\geq 8c ab^{1/k}}$ implies $t \gtrsim\frac{m^{2k}}{ a^k b^k}$.
\label{prop:cycle:count}
\end{prop}
\begin{proof}
First of all, from Lemma~\ref{lmm:jiang:yepremyan}, the following always holds for any
graph with $N$ nodes, $M$ edges, and $T$ cycles of length $2k$:
\begin{align*}
    T \geq d (M/N)^{2k} - dc^{2k} N^{2},
\end{align*}
where $c$ and $d$ are the $k$-dependent parameters from Lemma~\ref{lmm:jiang:yepremyan}. This is
true because if $M \geq cN^{1+1/k}$, then we have at least $d(M/N)^{2k}$ cycles, and if
$M < cN^{1+1/k}$ the right hand side becomes negative.

Consider for a moment that $a$ is a multiple of $b$. We will drop this assumption later.
We first handle the divisible case by showing that, if $m\geq 4cab^{1/k}$, then
$t \gtrsim \frac{m^{2k}}{a^kb^k}$.

We begin by partitioning $A$ randomly into $r = a/b $ parts $A_1, A_2, \ldots, A_r$ of size
$b$. Let $G_i$ be the bipartite graph induced by $(A_i, B)$; let $m_i$ be the number of
edges in $G_i$, and $X_i$ the number of cycles of length $2k$ in $G_i$. Then, we have
\begin{align*}
    X_i \geq d\left(\frac{m_i}{2b}\right)^{2k} - dc^{2k} (2b)^2.
\end{align*}
Let $X = \sum_{i=1}^r X_i$ be the total number of cycles of length $2k$ in the partition,
then
\begin{align*}
\revision{X} \geq \sum_{i=1}^r \left[d\left(\frac{m_i}{2b}\right)^{2k} - dc^{2k} (2b)^2\right]
\geq r \cdot d\left(\frac{\sum_i m_i}{2rb}\right)^{2k} - rdc^{2k} (2b)^2
= d\frac{m^{2k}}{(2b)^{2k} r^{2k-1}} - 4 dc^{2k} ab.
\end{align*}
Given a cycle in $G$, the probability that it is counted in $X$ is $r\binom b k / \binom a k
\leq 1/r^{k-1}$ (there is a factor of $r$ because there are $r$ choices for which
part the cycle lands on). Thus,
\begin{align*}
    \frac{t}{r^{k-1}} \geq \E[X] \geq d\frac{m^{2k}}{(2b)^{2k} r^{2k-1}} - 4dc^{2k} ab
\end{align*}
which implies
\begin{align*}
    t \geq d\frac{m^{2k}}{(2b)^{2k} r^{k}} - 4dc^{2k} abr^{k-1}
    = d\frac{m^{2k}}{4^k a^{k}b^{k}} - 4dc^{2k} \frac{a^{k}}{b^{k-2}}
    \geq \frac{d}{2} \cdot \frac{m^{2k}}{4^k a^k b^k}
\end{align*}

Lastly, to treat the case where $a$ is not a multiple of $b$, we can add $b\cdot  \lceil a / b \rceil - a$ dummy vertices to $A$, which in the worst case could double the size of $A$, and apply the bound in the new graph.
\end{proof}

\revision{Recall the notation $E(A,B)$
defined in Section~\ref{subsec:basic:definitions},
for the set of edges between vertex sets $A$ and $B$ in a graph,
where $A$ and $B$ can overlap.}
Our new edge bounding lemma is the following:
\begin{lmm}
Let $G$ be an undirected graph with $t$ ${2k}$-cycles and let \(A,B\subseteq V(G)\).
Then, the number
of edges between $A$ and $B$ is at most
\begin{align}
|E(A,B)| \lesssim \left(|A| \cdot |B|^{1/k} + |B| \cdot |A|^{1/k}\right) + t^{1/(2k)} \cdot \sqrt{|A| \cdot |B|}
\end{align}
\label{lmm:edge:count:listing}
\end{lmm}
\begin{proof}
Let \(M:=|E(A,B)|\). First suppose that \(A\cap B=\emptyset\). Let \revision{\(H = (A \cup B, E(A,B))\)} be
the bipartite graph with parts \(A,B\) and edge set \(E(A,B)\). Every
cycle of length $2k$ in \(H\) is also a similar length cycle in \(G\), so \(H\) contains at most
\(t\) such cycles. If \(M\lesssim \left( |A||B|^{1/k}+|B||A|^{1/k} \right)\), we are done.
Otherwise, \(M\) is above the threshold of Proposition~\ref{prop:cycle:count};
hence \(t\geq \Omega(M^{2k}/(|A|^k|B|^k))\), and so
\(M\lesssim t^{1/(2k)}\sqrt{|A||B|}\). Thus the desired bound holds in the
case when $A$ and $B$ are disjoint.

For the general case, let \(I:=A\cap B\). We split the vertices of \(I\)
randomly between the two sides, obtaining disjoint sets \(A'\subseteq A\) and
\(B'\subseteq B\). Each edge of \(E(A,B)\) appears in \(E(A',B')\) with probability at least
\(1/2\). Therefore, for some choice of the split, \(|E(A',B')|\geq |E(A,B)|/2\). Thus,
\[
M
\lesssim
\left(|A||B|^{1/k}
+
|B||A|^{1/k}
\right)
+
t^{1/(2k)}\sqrt{|A||B|}.
\]
\end{proof}

Using this edge bound on the $\Phi$-norm of an adjacency matrix of a graph with bounded degree, we bound
the number of capped $k$-walks:

\begin{lmm}
\label{lmm:capped:k:walks:bound}
Let $G$ be a graph with maximum degree $\Delta$ and $t$ cycles of length $2k$, and $\bm A$ be its
adjacency matrix. Then,
\[ \phinorm{\bm A} \lesssim
        m^{\frac{1}{2k}} \Delta^{\frac 1 2 - \frac{1}{2k}} +
        t^{\frac{1}{2k}} \log \Delta.
\]
\label{lmm:phi:norm:bound:listing}
\end{lmm}
\begin{proof}
By Lemma~\ref{lmm:indicator-phi-norm}, it is sufficient to show that, for
any subset $S$ of vertices of $G$, we have
\begin{align*}
\phinorm{\bm A \bm 1_S}
    \lesssim
        m^{\frac{1}{2k}} \Delta^{\frac 1 2 - \frac{1}{2k}} + t^{\frac{1}{2k}} \log \Delta
    \cdot \sqrt{|S|}.
\end{align*}
Let $E(v,S)$ denote the set of edges between $v$ and $S$. Define
\begin{align}
    f(x) := |\{ v : |E(v,S)| \geq x \}|
\end{align}
Then, our goal is to bound
$\phinorm{\bm A\bm 1_S} = \sum_{x=1}^{\Delta} \sqrt{f(x)}$,
subject to the constraints we know about $f$.
We will write down the constraints later, but first we break this sum in to
dyadic intervals.
For $i \in \{0, 1, \ldots, \lceil \log \Delta \rceil\}$, define
\begin{align*}
    B_i & := \{ v : 2^i \leq |E(v,S)| < 2^{i+1} \} \\
    b_i & := |B_i|.
\end{align*}
Then, because $f(x)$ is non-increasing in $x$, we have
\begin{align*}
    \phinorm{\bm A \bm 1_S}
    &= \sum_{i=0}^{\lceil \log \Delta \rceil} \sum_{x=2^i}^{2^{i+1}-1} \sqrt{f(x)} \\
    \revision{(f(x) \leq f(2^i) \text{ for } x \in [2^i, 2^{i+1}-1])}
    &\leq \sum_{i=0}^{\lceil \log \Delta \rceil} (2^{i+1}-2^i) \sqrt{f(2^i)} \\
    &= \sum_{i=0}^{\lceil \log \Delta \rceil} 2^i \sqrt{f(2^i)} \\
    \revision{\text{(by definition of $b_j$ and $f(2^i)$)}}
    &= \sum_{i=0}^{\lceil \log \Delta \rceil} 2^i \sqrt{\sum_{j=i}^{\lceil \log \Delta \rceil} b_j} \\
    \revision{\text{(sub-additivity of the square root)}}
    & \leq \sum_{i=0}^{\lceil \log \Delta \rceil} 2^i \sum_{j=i}^{\lceil \log \Delta \rceil} \sqrt{b_j} \\
    &= \sum_{j=0}^{\lceil \log \Delta \rceil} \sqrt{b_j} \sum_{i=0}^{j} 2^i \\
    &\leq 2 \sum_{j=0}^{\lceil \log \Delta \rceil} 2^j \sqrt{b_j}
\end{align*}
So our next job is to bound the quantity $2^j \sqrt{b_j}$ for each $j$. To do this, fix a $j
\in \{ 0, 1, \ldots, \lceil \log \Delta \rceil\}$.
To shorten notations, let $B = B_j$, $b = b_j$, $s = |S|$, and $d = 2^j$. Then, we
have the following constraints:
\begin{align}
    |E(B, S)| &\lesssim \left( sb^{1/k} + bs^{1/k} \right) + t^{1/(2k)} \cdot \sqrt{sb}
    && \revision{\text{(from Lemma~\ref{lmm:edge:count:listing})}}\\
    |E(B, S)| &\lesssim s\Delta \\
    bd \lesssim |E(B, S)| &\leq m. \label{eqn:bd:e:m}
\end{align}

Define $x = bd^2/s$, then we have the following bounds on $x$:
\begin{align*}
    x &= \frac{bd^2}{s} = d \frac{bd}{s} \leq d \frac{s\Delta}{s} = d \Delta \\
    x &= \frac{bd^2}{s} = \frac{(bd)^2}{bs}\\
    &\lesssim \frac{\left (sb^{1/k} + bs^{1/k} + t^{1/(2k)} \cdot \sqrt{sb}\right )^2}{bs} \\
    &\lesssim \frac{(sb^{1/k})^2 + (bs^{1/k})^2 + (t^{1/(2k)} \cdot \sqrt{sb})^2}{bs} \\
    &= \left( s b^{2/k - 1} +  bs^{2/k - 1} \right) + t^{1/k}\\
    &= \left(\frac{bd^2}{x} b^{2/k - 1} + \left(\frac{b}{s}\right)^{1-2/k}b^{2/k}\right) + t^{1/k}\\
    &\lesssim \left(\frac{d^2b^{2/k}}{x} +  \left(\frac{\Delta}{d}\right)^{1-2/k} \left(\frac m d\right)^{2/k}\right) + t^{1/k}\\
    &= \left( \frac{d^2b^{2/k}}{x} + \frac{\Delta^{1-2/k} m^{2/k}}{d}\right) + t^{1/k}.
\end{align*}

Now we use the fact that $x \leq \frac P x + Q$ implies $x = O(P^{1/2} + Q)$ to obtain that
\begin{align*}
    x &\lesssim \sqrt{d^2b^{2/k}} + \frac{\Delta^{1-2/k} m^{2/k}}{d} + t^{1/k} \\
    &= d b^{1/k} + \frac{\Delta^{1-2/k} m^{2/k}}{d} + t^{1/k}\\
    &\revision{=}\ (db)^{1/k}d^{1-1/k} + \frac{\Delta^{1-2/k} m^{2/k}}{d} + t^{1/k}\\
    &\lesssim m^{1/k} d^{1-1/k} + \frac{\Delta^{1-2/k} m^{2/k}}{d} + t^{1/k},
\end{align*}
\revision{where the last inequality follows from~\eqref{eqn:bd:e:m}.}

Next, along with the fact that $x \lesssim d\Delta$, we infer
\begin{align*}
    x &\lesssim \min \left\{m^{1/k} d^{1-1/k} + \frac{\Delta^{1-2/k} m^{2/k}}{d} + t^{1/k}, d \Delta  \right\} \\
    &\lesssim m^{1/k} d^{1-1/k} + \min \left\{\frac{\Delta^{1-2/k} m^{2/k}}{d}, d \Delta  \right\} + t^{1/k}.
\end{align*}

This implies
\begin{align*}
    \frac{2^j \sqrt{b_j}}{\sqrt s} = \frac{d \sqrt{b}}{\sqrt{s}} =\sqrt x \lesssim
    m^{\frac{1}{2k}} 2^{j\frac{k-1}{2k}} + \min \left\{\frac{\Delta^{\frac{k-2}{2k}} m^{1/k}}{2^{j/2}}, 2^{j/2} \sqrt \Delta  \right\} + t^{\frac{1}{2k}}.
\end{align*}
We bound the sum over $j$ of each of the first two terms separately. For the first term we have
\begin{align*}
    m^{1/2k} \sum_{j=0}^{\lceil \log \Delta \rceil} 2^{j(\frac{k-1}{2k})}
    & \frac{2^{(\lceil \log \Delta \rceil + 1)(\frac{k-1}{2k})} - 1}{2^{\frac{k-1}{2k}} - 1} \\
    &\lesssim m^{1/2k} 2^{\lceil \log \Delta \rceil (\frac{k-1}{2k})} \\
    &\lesssim m^{1/2k} \Delta^{\frac{k-1}{2k}}.
\end{align*}
For the second term, we need to find the value of $j \in \{0, 1, \ldots, \lceil \log \Delta
\rceil\}$ that transitions between the two cases in the minimum. \revision{Note that
$j^\star \geq 0 $ and, let's assume for a second that $j^\star < \lceil \log \Delta
\rceil$}. Then,
\[
\frac{\Delta^{\frac{k-2}{2k}} m^{1/k}}{2^{j/2}} \leq 2^{j/2} \sqrt \Delta
\qquad \text{ if and only if } \qquad
j \geq j^* := \lceil \log_2 ( m/\Delta)^{1/k}  \rceil.
\]
Hence, let's break the sum into two parts:
\begin{align*}
    \sum_{j=0}^{\lceil \log \Delta \rceil} \min \left\{ \frac{\Delta^{\frac{k-2}{2k}} m^{1/k}}{2^{j/2}}, 2^{j/2} \sqrt \Delta  \right\}
    &= \sum_{j=0}^{j^*} 2^{j/2} \sqrt \Delta + \sum_{j=j^*+1}^{\lceil \log \Delta \rceil} \frac{\Delta^{\frac{k-2}{2k}} m^{1/k}}{2^{j/2}} \\
    &\lesssim 2^{j^*/2} \sqrt \Delta + \frac{\Delta^{\frac{k-2}{2k}} m^{1/k}}{2^{j^*/2}} \\
    &\lesssim (m/\Delta)^{1/(2k)} \sqrt \Delta + \frac{\Delta^{\frac{k-2}{2k}} m^{1/k}}{(m/\Delta)^{1/(2k)}} \\
    &\lesssim m^{1/(2k)} \Delta^{\frac{k-1}{2k}}.
\end{align*}

And thus, $\sum_j 2^j \sqrt{b_j}$ has an extra $\log \Delta$ factor only on the $t^{1/(2k)}$
term, yielding the desired bound. \revision{In the case that $j^\star \geq \lceil \log
\Delta \rceil$ then $m > \Delta^{k+1}$, we can bound the sum of interest $$\sum_{j=0}^{\lceil
\log \Delta \rceil} \min \left\{ \frac{\Delta^{\frac{k-2}{2k}} m^{1/k}}{2^{j/2}}, 2^{j/2}
\sqrt \Delta  \right\}$$ by $\Theta(\Delta)$, using only the second term inside the min
expresion, which is $ \lesssim m^{1/(2k)} \Delta^{\frac{k-1}{2k}}$, as before}.
\end{proof}

Next, we prove the analog of Lemma 1.7 of ~\cite{DBLP:conf/stoc/DahlgaardKS17} for the
listing problem. Recall the quantity $\kappa_k(\revision{G,\Delta})$ defined in
Section~\ref{subsec:color:coding}, which is the maximum number of capped $k$-walks starting
from a vertex in $G$.
\begin{lmm}
    Let $G$ be an undirected graph with maximum degree $\Delta$ and $t$ cycles of length $2k$.
    Then,
    \begin{align*}
        \kappa_k(\revision{G, \Delta}) \lesssim
            m^{\frac 3 2 - \frac{1}{2k}} \Delta^{\frac{(k-1)^2}{2k}} + m^{\frac{2k}{k+1}} (\log\Delta)^{\frac{2k(k-1)}{k+1}} + t
    \end{align*}
\end{lmm}
\begin{proof}

Let $\bm A_p$ be the adjacency matrix of the graph induced by the vertices in $\bar{V}_p$. Using Lemma~\ref{lmm:capped:k:walks:bound} we have
\begin{align*}
\kappa_k(\revision{G, \Delta})
    &\lesssim m \cdot \max_{p} \phinorm{\bm A_p}^{k-1}\\
    &\lesssim m \left(
        m^{\frac{1}{2k}} \Delta^{\frac 1 2 - \frac{1}{2k}} + t^{\frac{1}{2k}}
        \log \Delta
   \right)^{k-1} \\
    &\lesssim m \cdot
        (m^{\frac{k-1}{2k}} \Delta^{\frac{(k-1)^2}{2k}} + t^{\frac{k-1}{2k}} \log^{k-1}\Delta)\\
    &=
        m^{\frac{3k-1}{2k}} \Delta^{\frac{(k-1)^2}{2k}} + mt^{\frac{k-1}{2k}} \log^{k-1}\Delta\\
    (\text{Young's inequality})
    &\lesssim
        \left(
            m^{\frac 3 2 - \frac{1}{2k}} \Delta^{\frac{(k-1)^2}{2k}} + m^{\frac{2k}{k+1}} (\log\Delta)^{2k(k-1)/(k+1)} + t
        \right)
\end{align*}
\end{proof}

Finally, we prove the main result of this paper, that we can list $2k$-cycles in undirected graphs more efficiently than
what the AYZ-time can achieve:
\begin{thm}
Let $G$ be an undirected graph with $m$ edges and $n$ vertices.
Let $k\geq 3$. We can list all $2k$-cycles in $G$ in time
$\lesssim m^{\frac{2k^2-k+1}{k^2+1}}\log n  + t \log n $, where $t$ is the number
of $2k$-cycles.
\end{thm}
\begin{proof}
From the AYZ approach described in Section~\ref{subsec:td} and Section~\ref{subsec:color:coding},
call a vertex ``heavy'' if its degree is more than
$\Delta$. The number of heavy vertices is at most $2m/\Delta$; hence, the heavy cycles
can be listed in time $\tilde O(m^2/\Delta + t)$.
The cycles involving light vertices can be listed in time $\tilde O(\kappa_k(\revision{G,\Delta}) + t)$, which is
\[ \widetilde O(
    m^{\frac 3 2 - \frac{1}{2k}} \Delta^{\frac{(k-1)^2}{2k}} + m^{\frac{2k}{k+1}} (\log\Delta)^{\frac{2k(k-1)}{k+1}} + t
). \]
We balance the two terms by setting $\Delta = m^{\frac{k+1}{k^2+1}}$, yielding a runtime of
$\widetilde O(m^{\frac{2k^2-k+1}{k^2+1}} + t)$.
The poly-log factor is absorbed because
$m^{\frac{2k^2-k+1}{k^2+1}}$ dominates $m^{\frac{2k}{k+1}} (\log\Delta)^{\frac{2k(k-1)}{k+1}}$.
\end{proof}

\section{Conclusions}\label{sec:conclusions}
We presented the first improvement to the AYZ even cycle listing result in 30 years.
Building on the capped-$k$-walk framework of Dahlgaard, Knudsen, and
St{\"{o}}ckel~\cite{DBLP:conf/stoc/DahlgaardKS17} and the supersaturation result of Jiang
and Yepremyan~\cite{DBLP:journals/cpc/JiangY20}, we proved that all $2k$-cycles in an
undirected $m$-edge graph can be listed in time $O\!\left(m^{(2k^2-k+1)/(k^2+1)} +
t\right)$, where $t$ is the output size. For $k = 3$, our algorithm matches the exponent of
Vassilevska~Williams and Westover~\cite{DBLP:conf/innovations/WilliamsW25} while avoiding
large polylogarithmic factors, computer-assisted case analysis, and BFS-based subroutines.
Our algorithms are expressed as query plans over multiple tree decompositions, \revision{using only
relational operators, which may facilitate their implementation in database
systems.} The take away message of our paper is the following:

\begin{enumerate}
\item A proof of the supersaturation conjecture~\ref{conj:supersaturation}. gives the
optimal bound for all $k$ -- this was already known
by~\cite{DBLP:conf/innovations/WilliamsW25}, but we shave off a $\log^{k+1}n$ factor from
their reduction.
\item Our weak supersaturation theorem gives the state of the art bound for the currently
known general case, as a function of $k$
\end{enumerate}

Several open problems remain.

\paragraph{Tightness of the exponent.}
The exponent $(2k^2-k+1)/(k^2+1)$ is unlikely to be optimal for all $k \geq 3$.

\paragraph{More general conjunctive queries.}
Our techniques apply specifically to even cycle queries, which involve a single symmetric
binary relation (the edge relation of an undirected graph). A natural and important open
problem is to generalize these ideas to broader classes of conjunctive queries where some of
the relations are symmetric and there are a lot of self-joins. The symmetry of the edge
relation was exploited crucially in our analysis --- through the Bondy--Simonovits theorem
and the capped-$k$-walk argument --- and it is unclear how to extend these ideas to
asymmetric or higher-arity settings.

\paragraph{Supersaturation beyond graphs.}
The key technical ingredient in our work is a supersaturation result for bipartite graphs,
derived from the Jiang--Yepremyan theorem.
A deeper open question is whether analogous supersaturation results hold for hypergraph
patterns or, more generally, for relations beyond binary ones.
Such results would be a prerequisite for extending our approach to a richer class of
conjunctive queries, and may be of independent interest in extremal combinatorics.

\bibliographystyle{plain}
\bibliography{main}

\begin{thebibliography}{10}

\bibitem{DBLP:conf/stoc/AbboudBDN18}
Amir Abboud, Karl Bringmann, Holger Dell, and Jesper Nederlof.
\newblock More consequences of falsifying {SETH} and the orthogonal vectors
  conjecture.
\newblock In Ilias Diakonikolas, David Kempe, and Monika Henzinger, editors,
  {\em Proceedings of the 50th Annual {ACM} {SIGACT} Symposium on Theory of
  Computing, {STOC} 2018, Los Angeles, CA, USA, June 25-29, 2018}, pages
  253--266. {ACM}, 2018.

\bibitem{DBLP:conf/fsttcs/AbboudKLS23}
Amir Abboud, Seri Khoury, Oree Leibowitz, and Ron Safier.
\newblock Listing 4-cycles.
\newblock In Patricia Bouyer and Srikanth Srinivasan, editors, {\em 43rd
  {IARCS} Annual Conference on Foundations of Software Technology and
  Theoretical Computer Science, {FSTTCS} 2023, {IIIT} Hyderabad, Telangana,
  India, December 18-20, 2023}, LIPIcs, pages 25:1--25:16. Schloss Dagstuhl -
  Leibniz-Zentrum f{\"{u}}r Informatik, 2023.

\bibitem{DBLP:books/aw/AbiteboulHV95}
Serge Abiteboul, Richard Hull, and Victor Vianu.
\newblock {\em Foundations of Databases}.
\newblock Addison-Wesley, 1995.

\bibitem{DBLP:conf/pods/KhamisNR16}
Mahmoud {Abo Khamis}, Hung~Q. Ngo, and Atri Rudra.
\newblock {FAQ:} questions asked frequently.
\newblock In Tova Milo and Wang{-}Chiew Tan, editors, {\em Proceedings of the
  35th {ACM} {SIGMOD-SIGACT-SIGAI} Symposium on Principles of Database Systems,
  {PODS} 2016, San Francisco, CA, USA, June 26 - July 01, 2016}, pages 13--28.
  {ACM}, 2016.

\bibitem{theoretics:13722}
Mahmoud {Abo Khamis}, Hung~Q. Ngo, and Dan Suciu.
\newblock {PANDA: Query Evaluation in Submodular Width}.
\newblock {\em TheoretiCS}, Volume 4, Apr 2025.

\bibitem{MR599482}
Noga Alon.
\newblock On the number of subgraphs of prescribed type of graphs with a given
  number of edges.
\newblock {\em Israel J. Math.}, 38(1-2):116--130, 1981.

\bibitem{color-coding}
Noga Alon, Raphael Yuster, and Uri Zwick.
\newblock Color-coding.
\newblock {\em J. ACM}, 42(4):844--856, 1995.

\bibitem{AYZ97}
Noga Alon, Raphael Yuster, and Uri Zwick.
\newblock Finding and counting given length cycles.
\newblock {\em Algorithmica}, 17(3):209--223, 1997.

\bibitem{DBLP:journals/siamcomp/AtseriasGM13}
Albert Atserias, Martin Grohe, and D{\'{a}}niel Marx.
\newblock Size bounds and query plans for relational joins.
\newblock {\em {SIAM} J. Comput.}, 42(4):1737--1767, 2013.

\bibitem{bondy1974cycles}
John~Adrian Bondy and Mikl{\'o}s Simonovits.
\newblock Cycles of even length in graphs.
\newblock {\em Journal of Combinatorial Theory, Series B}, 16(2):97--105, 1974.

\bibitem{DBLP:conf/stoc/BringmannG25}
Karl Bringmann and Egor Gorbachev.
\newblock A fine-grained classification of subquadratic patterns for subgraph
  listing and friends.
\newblock In Michal Kouck{\'{y}} and Nikhil Bansal, editors, {\em Proceedings
  of the 57th Annual {ACM} Symposium on Theory of Computing, {STOC} 2025,
  Prague, Czechia, June 23-27, 2025}, pages 2145--2156. {ACM}, 2025.

\bibitem{MR859293}
F.~R.~K. Chung, R.~L. Graham, P.~Frankl, and J.~B. Shearer.
\newblock Some intersection theorems for ordered sets and graphs.
\newblock {\em J. Combin. Theory Ser. A}, 43(1):23--37, 1986.

\bibitem{DBLP:conf/stoc/DahlgaardKS17}
S{\o}ren Dahlgaard, Mathias B{\ae}k~Tejs Knudsen, and Morten St{\"{o}}ckel.
\newblock Finding even cycles faster via capped k-walks.
\newblock In Hamed Hatami, Pierre McKenzie, and Valerie King, editors, {\em
  Proceedings of the 49th Annual {ACM} {SIGACT} Symposium on Theory of
  Computing, {STOC} 2017, Montreal, QC, Canada, June 19-23, 2017}, pages
  112--120. {ACM}, 2017.

\bibitem{Diestel2017}
Reinhard Diestel.
\newblock {\em Graph Theory}, volume 173 of {\em Graduate Texts in
  Mathematics}.
\newblock Springer, 5th edition, 2017.

\bibitem{ErdosSimonovits1983}
Paul Erd{\H{o}}s and Mikl{\'o}s Simonovits.
\newblock Supersaturated graphs and hypergraphs.
\newblock {\em Combinatorica}, 3(2):181--192, 1983.

\bibitem{MR2104047}
Ehud Friedgut.
\newblock Hypergraphs, entropy, and inequalities.
\newblock {\em Amer. Math. Monthly}, 111(9):749--760, 2004.

\bibitem{MR1639767}
Ehud Friedgut and Jeff Kahn.
\newblock On the number of copies of one hypergraph in another.
\newblock {\em Israel J. Math.}, 105:251--256, 1998.

\bibitem{DBLP:journals/comgeo/GajentaanO12}
Anka Gajentaan and Mark~H. Overmars.
\newblock On a class of o(n\({}^{\mbox{2}}\)) problems in computational
  geometry.
\newblock {\em Comput. Geom.}, 45(4):140--152, 2012.

\bibitem{DBLP:conf/pods/GottlobGLS16}
Georg Gottlob, Gianluigi Greco, Nicola Leone, and Francesco Scarcello.
\newblock Hypertree decompositions: Questions and answers.
\newblock In Tova Milo and Wang{-}Chiew Tan, editors, {\em Proceedings of the
  35th {ACM} {SIGMOD-SIGACT-SIGAI} Symposium on Principles of Database Systems,
  {PODS} 2016, San Francisco, CA, USA, June 26 - July 01, 2016}, pages 57--74.
  {ACM}, 2016.

\bibitem{DBLP:conf/soda/GroheM06}
Martin Grohe and D{\'{a}}niel Marx.
\newblock Constraint solving via fractional edge covers.
\newblock In {\em Proceedings of the Seventeenth Annual {ACM-SIAM} Symposium on
  Discrete Algorithms, {SODA} 2006, Miami, Florida, USA, January 22-26, 2006},
  pages 289--298. {ACM} Press, 2006.

\bibitem{DBLP:journals/talg/GroheM14}
Martin Grohe and D{\'{a}}niel Marx.
\newblock Constraint solving via fractional edge covers.
\newblock {\em {ACM} Trans. Algorithms}, 11(1):4:1--4:20, 2014.

\bibitem{DBLP:conf/stoc/Itai77}
Alon Itai and Michael Rodeh.
\newblock Finding a minimum circuit in a graph.
\newblock In John~E. Hopcroft, Emily~P. Friedman, and Michael~A. Harrison,
  editors, {\em Proceedings of the 9th Annual {ACM} Symposium on Theory of
  Computing, May 4-6, 1977, Boulder, Colorado, {USA}}, pages 1--10. {ACM},
  1977.

\bibitem{DBLP:journals/cpc/JiangY20}
Tao Jiang and Liana Yepremyan.
\newblock Supersaturation of even linear cycles in linear hypergraphs.
\newblock {\em Comb. Probab. Comput.}, 29(5):698--721, 2020.

\bibitem{DBLP:conf/sosa/JinWZ24}
Ce~Jin, Virginia {Vassilevska Williams}, and Renfei Zhou.
\newblock Listing 6-cycles.
\newblock In Merav Parter and Seth Pettie, editors, {\em 2024 Symposium on
  Simplicity in Algorithms, {SOSA} 2024, Alexandria, VA, USA, January 8-10,
  2024}, pages 19--27. {SIAM}, 2024.

\bibitem{DBLP:conf/stoc/JinX23}
Ce~Jin and Yinzhan Xu.
\newblock Removing additive structure in 3sum-based reductions.
\newblock In Barna Saha and Rocco~A. Servedio, editors, {\em Proceedings of the
  55th Annual {ACM} Symposium on Theory of Computing, {STOC} 2023, Orlando, FL,
  USA, June 20-23, 2023}, pages 405--418. {ACM}, 2023.

\bibitem{DBLP:conf/icdt/KhamisNOS19}
Mahmoud~Abo Khamis, Hung~Q. Ngo, Dan Olteanu, and Dan Suciu.
\newblock Boolean tensor decomposition for conjunctive queries with negation.
\newblock In Pablo Barcel{\'{o}} and Marco Calautti, editors, {\em 22nd
  International Conference on Database Theory, {ICDT} 2019, Lisbon, Portugal,
  March 26-28, 2019}, LIPIcs, pages 21:1--21:19. Schloss Dagstuhl -
  Leibniz-Zentrum f{\"{u}}r Informatik, 2019.

\bibitem{DBLP:conf/innovations/Lincoln020}
Andrea Lincoln and Nikhil Vyas.
\newblock Algorithms and lower bounds for cycles and walks: Small space and
  sparse graphs.
\newblock In Thomas Vidick, editor, {\em 11th Innovations in Theoretical
  Computer Science Conference, {ITCS} 2020, Seattle, Washington, USA, January
  12-14, 2020}, LIPIcs, pages 11:1--11:17. Schloss Dagstuhl - Leibniz-Zentrum
  f{\"{u}}r Informatik, 2020.

\bibitem{DBLP:journals/jacm/Marx13}
D{\'{a}}niel Marx.
\newblock {Tractable Hypergraph Properties for Constraint Satisfaction and
  Conjunctive Queries}.
\newblock {\em J. {ACM}}, 60(6):42:1--42:51, 2013.

\bibitem{10.1145/3196959.3196990}
Hung~Q. Ngo.
\newblock Worst-case optimal join algorithms: Techniques, results, and open
  problems.
\newblock In {\em Proceedings of the 37th ACM SIGMOD-SIGACT-SIGAI Symposium on
  Principles of Database Systems}, PODS '18, page 111–124, New York, NY, USA,
  2018. Association for Computing Machinery.

\bibitem{DBLP:conf/pods/NgoPRR12}
Hung~Q. Ngo, Ely Porat, Christopher R{\'{e}}, and Atri Rudra.
\newblock Worst-case optimal join algorithms: [extended abstract].
\newblock In Michael Benedikt, Markus Kr{\"{o}}tzsch, and Maurizio Lenzerini,
  editors, {\em Proceedings of the 31st {ACM} {SIGMOD-SIGACT-SIGART} Symposium
  on Principles of Database Systems, {PODS} 2012, Scottsdale, AZ, USA, May
  20-24, 2012}, pages 37--48. {ACM}, 2012.

\bibitem{DBLP:journals/sigmod/NgoRR13}
Hung~Q. Ngo, Christopher R{\'{e}}, and Atri Rudra.
\newblock Skew strikes back: new developments in the theory of join algorithms.
\newblock {\em {SIGMOD} Rec.}, 42(4):5--16, 2013.

\bibitem{DBLP:conf/innovations/WilliamsW25}
Virginia {Vassilevska Williams} and Alek Westover.
\newblock Listing 6-cycles in sparse graphs.
\newblock In Raghu Meka, editor, {\em 16th Innovations in Theoretical Computer
  Science Conference, {ITCS} 2025, January 7-10, 2025, Columbia University, New
  York, NY, {USA}}, volume 325 of {\em LIPIcs}, pages 92:1--92:21. Schloss
  Dagstuhl - Leibniz-Zentrum f{\"{u}}r Informatik, 2025.

\bibitem{DBLP:conf/icdt/Veldhuizen14}
Todd~L. Veldhuizen.
\newblock Triejoin: {A} simple, worst-case optimal join algorithm.
\newblock In Nicole Schweikardt, Vassilis Christophides, and Vincent Leroy,
  editors, {\em Proc. 17th International Conference on Database Theory (ICDT),
  Athens, Greece, March 24-28, 2014}, pages 96--106. OpenProceedings.org, 2014.

\bibitem{VassilevskaWilliams2018ICM}
Virginia~Vassilevska Williams.
\newblock On some fine-grained questions in algorithms and complexity.
\newblock In Boyan Sirakov, Paulo~Ney de~Souza, and Marcelo Viana, editors,
  {\em Proceedings of the International Congress of Mathematicians (ICM 2018)},
  volume 3: Invited Lectures, pages 3447--3487, Rio de Janeiro, Brazil, 2018.
  World Scientific.

\bibitem{DBLP:conf/vldb/Yannakakis81}
Mihalis Yannakakis.
\newblock Algorithms for acyclic database schemes.
\newblock In {\em Very Large Data Bases, 7th International Conference,
  September 9-11, 1981, Cannes, France, Proceedings}, pages 82--94. {IEEE}
  Computer Society, 1981.

\bibitem{DBLP:journals/siamdm/YusterZ97}
Raphael Yuster and Uri Zwick.
\newblock Finding even cycles even faster.
\newblock {\em {SIAM} J. Discret. Math.}, 10(2):209--222, 1997.

\end{thebibliography}
\appendix
\revision{

\section{Missing proofs from Section~\ref{sec:background}}

The following proof is adapted from~\cite{DBLP:conf/stoc/DahlgaardKS17}.
\begin{proof}[Proof of Lemma~\ref{lmm:capped-walks}]
Let $V_i$ be the set of vertices with degree in the interval $(2^{i-1}, 2^i]$ for $i \leq
\log \Delta$. Let $G_i$ denote the subgraph of $G$ induced by vertices of degree at most
$2^i$, and $\bm A_i$ be its adjacency matrix. A capped-$k$ walk in $G$ starting at a vertex
in $V_i$ is a $k$-walk in $G_i$ starting at a vertex in $V_i$.

Let $V(G_i)$ denote the set of vertices of $G_i$, and $\bm 1_{V_i}$ denote the indictor
vector of the set $V_i \subseteq V(G_i)$; namely $(\bm 1_{V_i})_v = 1$ if $v \in V_i$ and
$0$ if $v \in V(G_i) \setminus V_i$. Then, $\bm A_i^k \bm 1_{V_i}$ is a vector whose
$V_i$-entries count the number of $k$-walks in $G_i$ starting at each vertex in $V_i$. The
rest of the entries are $0$s. The $1$-norm of this vector, $\norm{\bm A_i^k \bm 1_{V_i}}_1$,
counts the total number of $k$-walks in $G_i$ starting at vertices in $V_i$.

For convenience, let $\bm w := \bm A_i^k \bm 1_{V_i}$. Then, Dahlgaard et al.~\cite{DBLP:conf/stoc/DahlgaardKS17} show that
observed that we can bound $\norm{\bm w}_1$ as follows, and this is where the $\Phi$-norm
comes into play.
Recall that the $0$-norm of a vector is the number of its non-zero entries, we have
\begin{align*}
\norm{\bm w}_1
    &:= \int_0^\infty |\{ v \in V_i : w_v \geq t \}| \, dt \\
    &= \int_0^\infty \sqrt{|\{ v \in V_i : w_v \geq t \}|}  \cdot \sqrt{|\{ v \in V_i : w_v \geq t \}|} \, dt \\
    &\leq \int_0^\infty \sqrt{|V_i|}  \cdot \sqrt{|\{ v \in V_i : w_v \geq t \}|} \, dt \\
    &= \sqrt{|V_i|} \cdot \phinorm{\bm A_i^k \bm 1_{V_i}}.
\end{align*}
Secondly, from the definition of the $\Phi$-norm of a matrix~\eqref{eqn:matrix:phi:norm}, we
can bound the magnification of every multiplication by $\bm A_i$ in terms of its
$\Phi$-norm. Specifically, for any vector $\bm x$, we have
\begin{align*}
  \phinorm{\bm A_i \bm x} \leq \phinorm{\bm A_i} \cdot \phinorm{\bm x}.
\end{align*}
which implies
\begin{align*}
  \phinorm{\bm A_i^k \bm x} \leq \phinorm{\bm A_i}^{k} \cdot \phinorm{\bm x}.
\end{align*}

Putting everything together,
\begin{align*}
\kappa_k(G, \Delta)
    &= \sum_{i} \kappa_k(V_i, G) \\
    & \leq \sum_{i} \kappa_k(V_i, G_i) \\
    &= \sum_{i} \norm{\bm A_i^k \bm 1_{V_i}}_1 \\
    &\leq \sum_{i} \sqrt{|V_i|} \cdot \phinorm{\bm A_i^k \bm 1_{\bm V_i}} \\
    &\leq \sum_i \phinorm{\bm A_i}^{k-1} \sqrt{|V_i|} \cdot \phinorm{\bm A_i \bm 1_{\bm V_i}} \\
    &\leq (\max_i \norm{\bm A_i}_\Phi)^{k-1}
        \sum_i \sqrt{|V_i|} \cdot \phinorm{\bm A_i \bm 1_{\bm V_i}}.
\end{align*}
Finally, we bound the inner sum:
$\bm A_i \bm 1_{\bm V_i} \leq \sum_{j \leq i} 2^j \bm 1_{V_j}$, and hence
\begin{align*}
\sum_i \sqrt{|V_i|} \cdot \phinorm{\bm A_i \bm 1_{\bm V_i}}
    & \leq \sum_i \sqrt{|V_i|} \cdot \sum_{j \leq i} 2^j \sqrt{|V_j|} \\
    &= \sum_i \sum_{j \leq i} 2^j \sqrt{|V_i|} \cdot \sqrt{|V_j|} \\
    &= \sum_i \sum_{j \leq i} \sqrt{2^i|V_i|} \cdot \sqrt{2^j|V_j|} \cdot 2^{(j-i)/2}\\
    &\leq \frac 1 2 \sum_i \sum_{j \leq i} (2^i|V_i| + 2^j|V_j|) \cdot 2^{(j-i)/2}\\
    &= \frac 1 2 \sum_i \sum_{j \leq i} 2^i|V_i| \cdot 2^{(j-i)/2}
    + \frac 1 2 \sum_i \sum_{j \leq i} 2^j|V_j| \cdot 2^{(j-i)/2} \\
    &\leq \frac 1 2 \sum_i 2^i |V_i| \sum_{j \leq i} 2^{(j-i)/2}
    + \frac 1 2 \sum_j 2^j|V_j| \sum_{i \geq j} 2^{(j-i)/2} \\
    &= O\left( \sum_i 2^i |V_i| \right) \\
    &= O(m)
\end{align*}
The last equality follows because every vertex in $V_i$ has degree at most $2^i$, and
the total degree of the graph is $2m$, which implies $\sum_i 2^i |V_i| = O(m)$.
\end{proof}

}

\end{document}